\documentclass[11pt]{article}
\usepackage{fullpage}
\usepackage{amsmath,bm}
\usepackage{moresize}
\usepackage{amssymb}
\usepackage{tikz-cd}
\usepackage{graphicx}
\usepackage{parskip}
\usepackage{enumitem}
\usepackage{amsthm}
\theoremstyle{definition}
\usepackage{listings}
\usepackage{color}
\usepackage{cancel}
\usepackage{framed}
\usepackage{relsize}
\usepackage{yfonts}
\usepackage{anyfontsize}
\usepackage{t1enc}
\usepackage[ruled,vlined]{algorithm2e}
\usepackage{algpseudocode}
\allowdisplaybreaks
\usepackage[utf8]{inputenc}
\usepackage[english]{babel}
\usepackage{lipsum}
\usepackage{mathtools}
\usepackage{verbatim}
\usepackage{multirow}
\usepackage{subfigure}
\usepackage[colorlinks=true]{hyperref}

\usepackage{cleveref}
\crefformat{section}{\S#2#1#3} 
\crefformat{subsection}{\S#2#1#3}
\crefformat{subsubsection}{\S#2#1#3}

\DeclarePairedDelimiter\ceil{\lceil}{\rceil}
\DeclarePairedDelimiter\floor{\lfloor}{\rfloor}

\hypersetup{final}

\newtheorem{Thm}{Theorem}[section]

\newtheorem{Lemma}[Thm]{Lemma}
\newtheorem{Def}[Thm]{Definition}
\newtheorem{Prop}[Thm]{Proposition}
\newtheorem{Claim}[Thm]{Claim}

\newtheorem{Fact}[Thm]{Fact}

\numberwithin{equation}{section}

\newcommand{\N}{\mathbb{N}}
\newcommand{\Z}{\mathbb{Z}}
\newcommand{\Q}{\mathbb{Q}}
\newcommand{\R}{\mathbb{R}}
\newcommand{\C}{\mathbb{C}}

\newcommand{\F}{\mathbb{F}}

\newcommand{\eq}{\equiv}

\newcommand{\A}{\mathcal{A}}

\newcommand{\Cc}{\mathcal{C}}

\newcommand{\I}{\mathcal{I}}

\newcommand{\W}{\mathcal{W}}

\newcommand{\ab}{\bold{a}}
\newcommand{\bb}{\bold{b}}

\newcommand{\ib}{\bold{i}}

\newcommand{\xb}{\bold{x}}
\newcommand{\yb}{\bold{y}}
\newcommand{\zb}{\bold{z}}

\newcommand{\Mb}{\bold{M}}

\newcommand{\Zb}{\bold{Z}}

\newcommand{\Xb}{\bold{X}}
\newcommand{\Yb}{\bold{Y}}
\newcommand{\Ib}{\bold{I}}

\definecolor{LightBlue}{rgb}{0, 0.7, .7}
\definecolor{Green1}{rgb}{0.0, 0.5, 0.0}
\definecolor{green}{rgb}{0.0, 0.42, 0.24}
\definecolor{byzantine}{rgb}{0.74, 0.2, 0.64}

\algrenewcommand{\Return}{\State\algorithmicreturn~}

\newcommand{\ind}{\text{\color{white}.$\quad$}}

\newcommand{\scheme}[1]{\ensuremath{\text{#1}}}

\newcommand{\Alg}{\scheme{A}}  
\newcommand{\rs}{\scheme{RS}}  
\newcommand{\frs}{\scheme{FRS}}  
\newcommand{\Enc}{\scheme{Enc}}  
\newcommand{\Der}{\scheme{Der}}  

\makeatletter
\def\legendre@dash#1#2{\hb@xt@#1{%
  \kern-#2\p@
  \cleaders\hbox{\kern.5\p@
    \vrule\@height.2\p@\@depth.2\p@\@width\p@
    \kern.5\p@}\hfil
  \kern-#2\p@
  }}
\def\@legendre#1#2#3#4#5{\mathopen{}\left(
  \sbox\z@{$\genfrac{}{}{0pt}{#1}{#3#4}{#3#5}$}%
  \dimen@=\wd\z@
  \kern-\p@\vcenter{\box0}\kern-\dimen@\vcenter{\legendre@dash\dimen@{#2}}\kern-\p@
  \right)\mathclose{}}
\newcommand\legendre[2]{\mathchoice
  {\@legendre{0}{1}{}{#1}{#2}}
  {\@legendre{1}{.5}{\vphantom{1}}{#1}{#2}}
  {\@legendre{2}{0}{\vphantom{1}}{#1}{#2}}
  {\@legendre{3}{0}{\vphantom{1}}{#1}{#2}}
}
\def\dlegendre{\@legendre{0}{1}{}}
\def\tlegendre{\@legendre{1}{0.5}{\vphantom{1}}}
\makeatother

\title{Beyond the Guruswami-Sudan (and Parvaresh–Vardy) Radii:
Folded Reed-Solomon, Multiplicity and
Derivative Codes}
\author{Neophytos Charalambides}
\date{December 10, 2019}

\begin{document}

\maketitle

\begin{abstract}
The classical family of Reed-Solomon codes consist of evaluations of polynomials over the finite field $\F_q$ of degree less than $k$, at $n$ distinct field elements. These are arguably the most widely used and studied codes, as they have both erasure and error-correction capabilities, among many others nice properties. In this survey we study closely related codes, folded Reed-Solomon codes, which are the first constructive codes to achieve the list decoding capacity. We then study two more codes which also have this feature, \textit{multiplicity codes} and \textit{derivative codes}. Our focus for the most part are the list decoding algorithms of these codes, though we also look into the local decodability of multiplicity codes.
\end{abstract}

\maketitle

\section{Introduction}
\label{Intro}

Communicating information is ubiquitous in modern technologies and every day interactions between people and corporations. Communication is achieved by encoding a message of length $k$ to a \textit{codeword} of length $n$ over alphabets, which is sent though a channel. The codeword may be corrupted in a subset of up to $\delta=pn$ symbols, for some $p\in(0,1)$. The purpose of encoding the original message is to have reliable communication over the channel, which means that the fraction of corrupted symbols may be restored; in order to retrieve the original message. A code $\Cc$ over an alphabet $\Sigma$ is a structured subset of $\Sigma^n$ for which such recoveries are possible, as long as there are no more than $pn$ corrupted symbols. The rate of $\Cc$ is defined as $R=\frac{\log|\Cc|}{n\log|\Sigma|}=\frac{k}{n}$.\\
$\ind$ A basic trade-off in this setting, is the one between rate $R$ and error fraction $p$; or equivalently between $R$ and the relative distance $\delta$. Clearly, $R\leq 1-p$. If we relax the decoding we require from unique, to listing a set of codewords which contain the correct codeword, this rate is asymptotically met. That is, there exist codes of rate $R=1-p-o(1)$ which are $p$-list-decodable. We refer to $1-R$ as the \textit{list decoding capacity}, which coincides with the fraction of errors we can correct, and is the optimal limit. Surprisingly, this is \textit{twice} the fraction of errors that one could decode when requiring unique decoding \cite{Gur07}! Though the above argument is non-constructive, \textit{folded Reed-Solomon codes} achieve list decoding from an error rate approaching $1-R$, with a polynomial time decoding algorithm \cite{GR08}. We present these codes in \cref{FRS_sec}, along with the original ideas and results from \cite{GR08}. We then describe two more completely different list decoding procedures for these codes, a linear-algebraic approach in \cref{LinAlg_sec}, and one based on Hensel-lifting \cref{HL_sec}.\\
$\ind$ In \cref{LDCs_Mult_codes_sec} we shift our focus to study another recent family of codes, \textit{multiplicity codes} \cite{KSY10}. Multiplicity error-correcting codes are locally decodable codes which have efficient local decoding algorithms, with rate approaching 1 and a low number of queries. They are based on evaluating multivariate polynomials and their derivatives. Finally, in \cref{Der_Codes_sec} we delve into a closely related family of codes, \textit{derivative codes}. These are simpler and more natural codes which relate to folded Reed-Solomon codes, putting together a lot of the ideas we will see throughout this survey.

\section{Folded Reed-Solomon Codes}
\label{FRS_sec}

Recall that a Reed-Solomon code $\rs_q[n,k]$ (RS) over $\F_q$, is the encoding of polynomials of degree at most $k-1$ which represents our message, over the \textit{defining set of points} $\A=\{\alpha_1,\cdots,\alpha_n\}\subset\F_q$
$$ \rs_q[n,k]=\Big\{\big[f(\alpha_1),f(\alpha_2),\cdots,f(\alpha_n)\big] \ \Big| \ f(X)\in\F_q[X] \text{ of degree }\leq k-1 \Big\} $$
Typically $n=|\F_q|-1=q-1$ and $\alpha_i=\alpha^{i}$ for all $i\in\N_n\coloneqq\{1,2,\cdots,n\}$, where $\alpha$ is some primitive element in $\F_q$. The encoding of the message $\vec{m}=\left[m_0,\cdots,m_{k-1}\right]\mapsto m(X)=\sum_{i=0}^{k-1}m_i X^i\in\F_q[X]$ is defined by the evaluation mapping
$$ \Enc(m)=\big[m(\alpha_1),m(\alpha_2),\cdots,m(\alpha_n)\big]\in\F_q. $$
$\ind$ One difficulty with RS codes is that we need to be able to correct \textit{any} pattern of $\floor{\frac{n-k+1}{2}}$. Guruswami and Rudra \cite{GR08} address this problem, by ``bundling'' parts of the codewords together, which considerably decreases the error pattern we have to handle. What they define as \textit{folded Reed-Solomon} codes (FRS), are in fact exactly RS codes, but viewed as a code over a larger alphabet by careful bundling of codeword symbols. Informally, an $m$-FRS code is a RS code over $\F_q$, where $m$ consecutive positions in the RS code are identified with an element in $\F_{q^m}$. That is, the columns of the encoded matrix (\ref{frs_eq}) may each be considered as an element in $\F_{q^m}$ --- \textbf{folding} a vector in $\F_q^m$ to an element in $\F_{q^m}$. We point out that the term \textit{folded Reed-Solomon} was first introduced in \cite{Kra03} to correct burst errors, though the folding operation is slightly different to what we are considering.

\begin{Def}
\textit{Consider $\Sigma=\F_q$, its nonzero elements $\{1,\gamma,\cdots,\gamma^{n-1}\}$, for $n=q-1$, and $\gamma$ a primitive element. Let $m$ be a positive factor of $n$; i.e. $n=N\cdot m$, and \textbf{degree parameter} $k\in\N_n$. The \textbf{$m$-folded Reed-Solomon code}} $\frs_q^{(m)}[k]$, \textit{is a code over $\F_q^m\cong\F_{q^m}$, that encodes a polynomial $f(X)\in\F_q[X]$ of degree $k-1$}
\begin{equation}
\label{frs_eq}
  {\footnotesize \Enc\left(f(X)\right) = \left( \begin{bmatrix} f(1) \\ f(\gamma) \\ \vdots \\ f(\gamma^{m-1}) \end{bmatrix}, \begin{bmatrix} f(\gamma^{m}) \\ f(\gamma^{m+1}) \\ \vdots \\ f(\gamma^{2m-1}) \end{bmatrix}, \cdots, \begin{bmatrix} f(\gamma^{n-m}) \\ f(\gamma^{n-m+1}) \\ \vdots \\ f(\gamma^{n-1}) \end{bmatrix} \right) \cong \begin{bmatrix} f(1) & \cdots & f(\gamma^{n-m}) \\ f(\gamma) & \cdots & f(\gamma^{n-m+1}) \\ \vdots & \ddots & \vdots \\ f(\gamma^{m-1}) & \cdots & f(\gamma^{n-1}) \end{bmatrix}}.
\end{equation}
\end{Def}

\begin{Prop}
\textit{The} FRS \textit{(nonlinear)
code over $\F_q^m$ defined above, has block length $N$, rate $R=\frac{k}{n}=\frac{k}{Nm}$, and minimum distance} $d_\text{min}=N-\ceil{k/m}+1\simeq (1-R)N$.
\end{Prop}

$\ind$ Suppose a FRS codeword was transmitted, and a received (potentially corrupted) string
\begin{equation}
\label{rec_matrix_cw} 
  \yb = \begin{pmatrix} y_{11} & y_{12} & \cdots & y_{1n} \\ y_{21} & y_{22} & \cdots & y_{2N} \\ \vdots & \vdots & \ddots & \vdots \\ y_{m1} & y_{m2} & \cdots & y_{mN} \end{pmatrix} \in \F_q^{m\times N}
\end{equation}
$\yb\in(\F_q^m)^N$ was received, which we view as a matrix in $\F_q^{m\times N}$. The goal is to recover a list of all polynomials in $\F_q[X]$ of degree at most $k-1$, whose encoding (\ref{frs_eq}) agrees with $\yb$ in at least $t$ columns, for some \textbf{agreement parameter} $t$. Ideally, we would like $t$ to be as small as possible, as this corresponds to list decoding up to $n-t$ errors. We know how to list decode a RS code up to $1-\sqrt{R}$ \cite{GS98} in $\text{poly}(n,q)$ time, so by simply unfolding and treating it as a regular RS codes, we can solve this for $t\geq \sqrt{R}N$. A crucial difference in FRS is its additional structure: when a column is correct, we know the correct values of \textit{all} $m$ values in the column.\\
$\ind$ Decoding these codes is similar in spirit to list decoding of RS. The gain comes from interpolating in more than two dimensions. That is, we seek a higher dimensional analog of the identity $Q(X,f(X))=0$ from the RS case, a low-degree nonzero polynomial which is interpolated through the data. The essence is to argue that this identity suffices to retrieve a small list of possibilities efficiently. The main steps in this higher dimensional version of the Berlekamp-Welch algorithm \cite{BW86}, are \textit{interpolation} and \textit{root-finding}. We now present a fundamental result of \cite{GR08}.

\begin{Thm}[\cite{Gur11}]
\label{list_dec_thm_GR08}
\textit{For every integer $s,1\leq s\leq m$ and any constant $\delta>0$, there is a list decoding algorithm for} $\frs_q^{(m)}[n,k]$ \textit{that list decodes from up to $e$ errors, as long as}
$$ e\leq N-(1+\delta)\frac{\big(k^sN(m-s+1)\big)^{1/(s+1)}}{m-s+1} $$
\textit{where $N=\frac{n}{m}$ is  the code block length. The algorithm runs in $\big(O_\delta(q)\big)^{O(s)}$ time, and outputs a list of size at most $q^{s-1}$}. 
\end{Thm}

$\ind$ The fraction of errors corrected by this algorithm as a function of the rate $R$ is
\begin{equation}
\label{fract_corr_errs}
  1-(1+\delta)\left(\frac{mR}{m-s+1}\right)^{s/(s+1)}\overset{\natural}{\simeq}1-(1-\varepsilon)\left(\frac{R}{1-\varepsilon+\varepsilon^2}\right)^{1/(1+\varepsilon)}\geq 1-R-\varepsilon
\end{equation}
where in $\natural$ we pick $\delta\simeq\varepsilon$, $s\simeq1/\varepsilon$ and $m\simeq s^2$. The decoding complexity and list-size are $\simeq q^{O(1/\varepsilon)}$.\\ $\ind$ Another important thing to note here is that we consider a fraction $\simeq 1-\left(\frac{mR}{m-s+1}\right)^{s/(s+1)}$ of errors, where for large enough $m$ we get $\simeq 1-\sqrt[s+1]{R^s}$. The second term is precisely the geometric mean of $1,R,\cdots,R$. This is analogous to the improvement achieved in RS codes $(s=1)$, where the agreement required between the received string and codeword was reduced from $\frac{1+R}{2}$ (arithmetic mean) to $\sqrt{R}$ (geometric mean). The corresponding radius was consequently improved from $\frac{1-R}{2}$ to $1-\sqrt{R}$, which is always better by the AM-GM inequality.\\
$\ind$ The main gain of the \textit{folding operation} on RS codes, is that we can construct list decodable FRS codes up to radius roughly $1-\sqrt[s+1]{R^s}$, for any $s\in\Z_+$. By selecting $s$ large enough, we can get within any desired $\varepsilon$ from capacity, attaining list decodability up to fraction $1-R-\varepsilon$ of errors. Moreover, list  decoding capacity was achieved over large alphabets \cite{GR08}, as $\lim_{s\to\infty}\left\{1-\sqrt[s+1]{R^s}\right\}=1-R$, though there was room for improvement with respect to some of the parameters.

\subsection{Interpolation step}
\label{int_step_subsec_frs}

The list decoding algorithm (\cite{Gur10},\cite{Vad12},\cite{GRS19} --- modified version of the original algorithm) first interpolates a linear polynomial $Q(X,\Yb)=Q(X,Y_1,\cdots,Y_s)$ of degree 1 in the $Y_i$'s through certain $(s+1)$-tuples, where $\Yb$ denotes the formal variables $Y_1,\cdots,Y_s$. Given a $\yb\in\F_q^{m\times N}$, we interpolate a \textit{nonzero} polynomial
\begin{equation}
\label{Q_interp}
  Q(X,\Yb) = A_0(X)+A_1(X)Y_1+\cdots+A_s(X)Y_s = A_0(X)+\sum\limits_{i=1}^sA_i(X)Y_i
\end{equation}
where deg$(A_i)\leq d$ for all $i\in\N_s$ and deg$(A_0)\leq d+k-1$, for a suitable degree parameter. The total number of monomials which appear in $Q$ with these restrictions is
$$ (d+1)s+d+k=(d+1)(s+1)+k-1\geq N(m-s+1)+s+1>N(m-s+1) $$
for $d$ chosen to be
\begin{equation}
\label{d_express}
  d=\left\lfloor\frac{N(m-s+1)-k+1}{s+1}\right\rfloor
\end{equation}
and $Q\in\F_q[X][\Yb]\cong\F_q[X,\Yb]=\F_q[X,Y_1,\cdots,Y_s]$ must satisfy the interpolation step
\begin{equation}
\label{int_cond_frs}
  Q\left(\gamma^{im+j},y_{im+j},y_{im+(j+1)},\cdots,y_{im+(j+s-1)}\right)=0
\end{equation}
for all $i=0,1,\cdots,N-1, \ j=0,\cdots,m-s$, which may be viewed as $N(m+s-1)$ constraints. Since we have more monomials than constraints, such a nonzero polynomial $Q$ exists, which can be found by solving a homogeneous linear system. This explains our choice of $d$. The following lemma gives a necessary algebraic condition which message polynomials $f(X)$ in our desired list must satisfy.

\begin{Lemma}
\label{alg_cond_msg}
\textit{If $f(X)\in\F_q[X]$ is a polynomial of degree at most $k-1$ whose} FRS \textit{encoding} (\ref{frs_eq}) \textit{agrees with $\yb$ in at least $t$ columns for $t\geq d+k$, and $s=m$, then}
\begin{equation}
\label{s_list_cond}
  Q(X,f(X),f(\gamma X),\cdots,f(\gamma^{s-1} X))=0.
\end{equation}
\textit{We refer to the polynomials $f(X)\in\F_q[X]$ satisfying (\ref{s_list_cond}) as \textbf{$\Yb$-root of $Q$}}.
\end{Lemma}

\begin{proof}
Define $R(X)\coloneqq Q(X,f(X),f(\gamma X),\cdots,f(\gamma^{s-1} X))$, for which deg$(R)\leq d+k-1$. If $\Enc(f(X))$ agrees with $\yb$ in the $i^{th}$ column; then $f(\gamma^{im+\iota})=y_{im+1+\iota}$, for $\iota\in\{0,1,\cdots,m-1\}$ and $i$'s as defined for (\ref{int_cond_frs}). Together with condition (\ref{int_cond_frs}) the assumption that $m=s$, this implies that
$$ R(\gamma^{is})=Q(\gamma^{is},f(\gamma^{is}),f(\gamma^{is+1}),\cdots,f(\gamma^{is+s-1}))=0. $$
It follows that $R(X)$ has at least $t\geq d+k$ zeros. Since deg$(R)\leq d+k-1$, by the fundamental theorem of algebra $R(X)\equiv0$.
\end{proof}

All in all, lemma \ref{alg_cond_msg} provides the correctness of the procedure we are describing.

\subsection{Root-finding step}
\label{root_find_subsec}

The second step of our decoding algorithm is an $(s+1)$-variate ``root-type'' problem:
\begin{changemargin}{0.6cm}{0.6cm}
  Given $Q(X,\Yb)\not\eq0$ with coefficients in $\F_q$, $\gamma\in\F_q$ a primitive element, and parameter $k<n=q-1$, find the list of all polynomials $f(X)$ of degree at most $k-1$ such that $Q(X,f(X),f(\gamma X),\cdots,f(\gamma^{s-1}X))=0$.
\end{changemargin}
The following algebraic lemma is an important step to solving this problem.

\begin{Lemma}
\label{alg_lemma}
\textit{For $\gamma\in\F_q$ a primitive element, we have}:
\begin{enumerate}
  \item \textit{The polynomial $E(X)\coloneqq X^{q-1}-\gamma$ is irreducible over $\F_q$}
  \item \textit{If} deg$(f)<q-1$, \textit{then $f(\gamma X)\equiv f(X)^q\bmod(X^{q-1}-\gamma)$}.
\end{enumerate}
\end{Lemma}

$\ind$ We now present how to list the polynomials $f(X)\in\F_q[X]$ of degree $\leq k-1$ for the trivariate case ($s=2$), to satisfy the condition $Q(X,f(X),f(\gamma X))\equiv0$. We then discuss how this is can be generalized to $s\geq3$.

\begin{Thm}
\label{root_find_triv}
\textit{Consider the finite field $\F_q$ with a primitive element $\gamma$, and $Q(X,Y_1,Y_2)\in\F_q[X,Y_1,Y_2]$ nonzero with} $\text{deg}_{Y1}(Q)\leq q-1$, \textit{along with an integer parameter $k<q$. There is a deterministic algorithm with runtime $poly(q)$, which outputs the list of all $f(X)$ of degree at most $k-1$, satisfying $Q(X,f(X),f(\gamma X))\equiv0$}.
\end{Thm}

\begin{proof}
We know by lemma \ref{alg_lemma} part 1 that $E(X)$ is irreducible. For $b\in\N_0$ such that $E(X)^b\parallel Q(X,Y_1,Y_2)$; i.e. $E(X)^{b+1}\nmid Q(X,Y_1,Y_2)$ while $E(X)^b\mid Q(X,Y_1,Y_2)$, factor out $E(X)^b$ to obtain $Q_0(X,Y_1,Y_2)=E(X)^{-b}\cdot Q(X,Y_1,Y_2)$. It is clear that $E(X)\nmid Q_0(X,Y_1,Y_2)$ , and that if $Q(X,f(X),f(\gamma X))=0$; then $Q_0(X,f(X),f(\gamma X))=0$.\\
$\ind$ We may therefore focus on $Q_0$ instead, which we view as a polynomial $T_0(Y_1,Y_2)\in\F_q[X][Y_1,Y_2]$, for which we reduce the coefficients of modulo $E(X)$ to get $T(Y_1,Y_2)\in\tilde{\F}[Y_1,Y_2]$, for $\tilde{\F}\coloneqq\F_q[X]/(E(X))\cong\F_{q^{q-1}}$. That is, the bivariate polynomial $T(Y_1,Y_2)$ is over the extension field $\tilde{\F}$. Further note that $T(Y_1,Y_2)\not\eq0$, since $E(X)\nmid Q_0(X,Y_1,Y_2)$.\\
$\ind$ From the second part of our lemma, it suffices to find all polynomials $f(X)$ of degree $\leq k-1$ satisfying $Q_0(X,f(X),f(X)^q)\equiv0\bmod E(X)$; i.e. $E(X)\mid Q_0(X,f(X),f(X)^q)$. This reduces to finding the elements $\Gamma\in\tilde{\F}$ satisfying $T(\Gamma,\Gamma^q)=0$. For the univariate polynomial $R_2(Y_1)\coloneqq T(Y_1,Y_1^q)$, this corresponds to finding its roots in $\tilde{\F}$ --- $R_2(\Gamma)=0$. To recap, $R_2(\Gamma)=0$ for $\Gamma\in\tilde{\F}$ has a correspondence with the coefficients of $f(x)\in\F_q[X]$ of $T_0(Y_1,Y_2)$, for which $Q_0(X,f(X),x(\gamma X))=0$.\\
$\ind$ Note that $R_2(Y_1)=0$ if and only if $(Y_2-Y_1^q)\mid T(Y_1,Y_2)$, which cannot happen as $\text{deg}_{Y_1}(T)<q$ (char$(\tilde{\F})\leq q$). Furthermore, deg$(R_2)\leq dq$ for $d$ the total degree of $Q(X,Y_1,Y_2)$. Since char$(\tilde{\F})\leq q$ and $[\tilde{\F}:\F_p]=[\tilde{\F}:\F_q]\cdot[\F_q:\F_p]=(q-1)\log q$, we have $[\tilde{\F}:\F_p]\leq q\log q$. Using Berlekamp's deterministic factorization algorithm, we can find all roots of $R_2(Y_1)$ in time poly$(d,q)$ \cite{Ber70},\cite{Ker09}. Each such root is retrieved as an element in $\tilde{\F}$, which corresponds to a polynomial $f(X) \in \F_q[X]$ of degree less than $ q-1$. Once we have this list, we reduce it by only outputting the polynomials $f(X)$ of degree at most $k-1$ satisfying $Q_0(X,f(X),f(\gamma X))=0$.
\end{proof}

$\ind$ In the case where $s\geq 3$ the ideas in the above proof still apply, where now we want to list all degree $k-1$ polynomials $f(X)\in\F_q[X]$ satisfying (\ref{s_list_cond}), i.e. the $\Yb$-roots of $Q$. By dividing $Q(X,\Yb)\in\F_q[X][\Yb]$ by $E(X)$ enough times, we can assume that not all coefficients in $\F_q[X]$ are divisible by $E(X)$. We then quotient out $E(X)$ to get a nonzero polynomial $T(\Yb)$ over $\tilde{\F}$. By lemma \ref{alg_lemma} part 2, $f(\gamma^jX)\eq f(X)^{q^j}\bmod E(X)$ for all $j\in\Z_+$. Our root-finding task is now reduced to finding all roots $\Gamma\in\tilde{\F}$ of $R_s(Y_1)\coloneqq T(Y_1,Y_1^q,Y_1^{q^2},\cdots,Y_1^{q^{{s-1}}})$. We further need to make the assumption that the total degree of $T$ is lees than $q$, to ensure that $R_s(Y_1)\not\eq0$. The degree of $R_s(Y_1)$ is at most $q^s$, which means all its roots can be found in $q^{O(s)}$ time.\\
$\ind$ With this approach, we retrieve a list of at most $q^s$ polynomials in poly$(q)$ time. With rate $R$ we achieve polynomial time list decoding up to a fraction $1-R-\varepsilon$ of errors for every $R$ and arbitrary $\varepsilon>0$, where the alphabet size is $n^{O(1/\varepsilon)}$ \cite{GR08}. The optimal trade-off between rate and error-correction capability is therefore attained algorithmically.

\section{Linear-Algebraic List Decoding of Folded Reed-Solomon Codes}
\label{LinAlg_sec}

In \cite{Gur11} a linear-algebra based analysis of a variant of the above algorithm was given, which avoids the computationally expensive root-finding step over $\tilde{\F}$. The main idea is to solve one linear system in place of the interpolation step, and another one to find a ``small'' subspace of candidate solutions. There is again the step of ``pruning'' the list of candidate solutions (in this case a subspace), but other than this, the linear-algebraic algorithm can be implemented in \textit{quadratic} time.\\
$\ind$ They key observation is that the candidate solutions to the algebraic equations we wish to solve form an affine subspace, of the full message space $\F_q^k$. This is precisely what allows us avoid the interpolation step, by solving instead a linear system. Furthermore, this implies that the exponential dependence in $s$
of the list-size bound $q^{s-1}$ mentioned earlier, was inherently because of the dimension of the interpolation, implying that the identity $f(\gamma^{s-1}X)=f(X)\gamma^{q^{s-1}}$ over $\tilde{\F}$ used in the generalization of theorem \ref{root_find_triv} was not crucial in finding the roots. However, this identity seems to be the only known way to bound the list-size when higher degrees are used in the interpolation.\\
$\ind$ For our new list decoding algorithm, we need to find all polynomials $f(X)\in\F_q[X]$ of degree at most $k-1$ that satisfy the system of linear equations
\begin{equation}
\label{A_pol_lin_al}
  \Lambda(X) \coloneqq A_0(X)+A_1(X)f(X)+A_2(X)f(\gamma X)\cdots+A_s(X)f(\gamma^{s-1}X)=0
\end{equation}
in the coefficients of $f(X)=\sum\limits_{i=0}^{k-1}f_iX^i\in\F_q[X]$. Fact \ref{fact_aff_space_one} gives an efficient algorithm to find a compact representation of all the solutions of (\ref{A_pol_lin_al}). Additionally, the proof of lemma \ref{aff_subsp_lem} exposes the simple structure of (\ref{A_pol_lin_al}), which can be used to find the basis of solutions in quadratic time.

\begin{Fact}
\label{fact_aff_space_one}
\textit{The solutions $(f_0,f_1,\cdots,f_{k-1})$ of} (\ref{A_pol_lin_al}) \textit{form an affine subspace of $\F_q^k$}.
\end{Fact}

\begin{Lemma}
\label{aff_subsp_lem}
\textit{If} ord$(\gamma)\geq k$ \textit{(which is met for $\gamma$ primitive and $k\leq q-1$), the affine subspace of solutions to} (\ref{A_pol_lin_al}) \textit{has dimension $\tilde{d}\leq s-1$. Further, one can compute using $O((Nm)^2)=O(n^2)$ operations over $\F_q$ a matrix $\Mb\in\F_q^{k\times \tilde{d}}$ (for some $\tilde{d}\leq s-1$) and a vector $\zb\in\F_q^k$, such that the solutions are contained in the affine space $\Mb\xb+\zb$ for $\xb\in\F_q^{\tilde{d}}$. Also, $\Mb$ can be assumed to have the identity matrix $\Ib_{\tilde{d}}$ as a submatrix (without any extra computation)}.
\end{Lemma}

\begin{proof}
By factoring out the common powers of $X$ that divide \{$A_i(X)\}_{i=0}^s$ from (\ref{Q_interp}), we can assume that $X\nmid A_{\iota}(X)$ for ate least one $\iota\in\{0,1,\cdots,s\}$ --- more specifically, has a nonzero constant term. Further, if $X\mid A_{i}(X)$ for all $i\in\N_s$; then $X\mid A_0(X)$, and we can take $\iota\in\N_s$.\\
$\ind$ For $0\leq i\leq s$ denote $A_i(X)=\sum\limits_{j=0}^{d+k-1}a_{i,j}X^j$. By the degree constraints on  $\{A_i(X)\}_{i=1}^s$ we have $a_{i,j}=0$ for all pairs $(i,i)\in \N_s\times\Z_{>d}$, but we still introduce these coefficients for notational convenience. Define
$$ B(X)\coloneqq a_{1,0}+a_{2,0}X+a_{3,0}X^2+\cdots+a_{s,0}X^{s-1} = \sum\limits_{i=0}^{s}a_{i,0}X^{i-1} $$
which corresponds to $A_0(X)$, for which deg$(B)\leq s-1$. Since $a_{\iota,0}\neq0$, it follows that $B(X)\not\eq0$.\\
$\ind$ It is clear that the constant term of $\Lambda(X)$ equals $\left(a_{0,0}+f_0\sum\limits_{i=0}^sa_{i,0}\right) = \left(a_{0,0}+B(1)f_0\right)$. Thus if $B(1)\neq0$, the coefficient $f_0$ is uniquely determined as $-a_{0,0}/B(1)$. If $B(1)=0$; then $a_{0,0}=0$, or there will be no solutions to (\ref{A_pol_lin_al}). In that case, we assign an arbitrary value in $\F_q$ to $f_0$.\\
$\ind$ The coefficient of $X^r$ of $\Lambda(X)$ equals
\begin{equation}
\label{coeffs_bl}
  \lambda_r = a_{0,r} + \left(\sum\limits_{j=0}^r f_{r-j}\left(\sum\limits_{i=1}^sa_{i,j}\gamma^{(i-1)(r-j)}\right)\right) = B(\gamma^{r})f_r+\left(\sum\limits_{l=0}^{r-1}b_l^{(r)}f_l\right)+a_{0,r}
\end{equation}
for some coefficients $b_l^{(r)}\in\F_q$, and (\ref{coeffs_bl}) must equal zero. Furthermore, if $B(\gamma^r)\neq0$; $f_r$ is an affine combination of $\{f_j\}_{j=0}^{r-1}$. In particular, $f_r$ is uniquely determined given the values of $\{f_j\}_{j=0}^{r-1}$.\\
$\ind$ The dimension of the space of solutions of (\ref{A_pol_lin_al}) is therefore at most $r$; $0\leq r<k$, for which $B(\gamma^r)=0$. By our assumption that ord$(\gamma)\geq k$, it follows that $\{\gamma^r\}_{r=0}^{k-1}$ are all distinct. Since $B(X)$ is a nonzero polynomial and deg$(B)\leq s-1$, we know that $B(\gamma^r)=0$ for at most $s-1$ values of $r$. This concludes the proof that the solution space is of dimension at most $s-1$. The claim regarding quadratic complexity and the structure of $\Mb$, follows from the fact that (\ref{coeffs_bl}) resembles a ``lower-triangular'' form, which can be solved in $O(n^2)$ with the \textit{back-substitution} method.
\end{proof}

$\ind$ We close off this section with some comments on the rest of the results from \cite{Gur11}. The algorithm we saw gives a quadratic runtime for the list decoder; except for the final step of pruning the subspace, which could take $O(q^s)$ time. The formal statement may be found in \cite{Gur11} theorem 7, which also relates to the discussion we had following theorem \ref{list_dec_thm_GR08}. Lastly, the author discusses a possible approach to improving the possible worst case list-size bound, by restricting the message coefficients $(f_0,\cdots,f_{k-1})$ to belong to a special subset $\mathcal{V}\subseteq\F_q^k$ which satisfy two conflicting demands; \textit{largeness} and (what he coined as) \textit{subspace-evasive} \cite{DL12},\cite{BAS14}.

\section{Hensel-Lifting for Folded Reed-Solomon Codes}
\label{HL_sec} 

In this section we present a third approach to the \textit{root-finding} problem. This is quite different to the approaches discussed in \cref{root_find_subsec}, \cref{LinAlg_sec}, and uses ideas developed in number theory; namely \textit{Hensel's (lifting) lemma} \ref{Hens_lem_simple}. We give a brief discussion on the importance of this in \cref{NT_subs}. By theorem \ref{list_dec_thm_GR08} we already have a polynomial time algorithm which predates the algorithm based on \textit{Hensel-lifting} \cite{Bra10},\cite{PB09}, though the fact that it works over the exponentially large finite field $\tilde{\F}\cong\F_{q^{q-1}}$, makes any practical implementations difficult and even more numerically unstable. This newer decoder, is also faster experimentally.

\begin{Lemma}[Simplest version, \cite{NZM91}]
\label{Hens_lem_simple}
\textit{Suppose that $f(x)\in\Z[x]$. If $f(a)\eq0\bmod p^j$ and $f'(a)\not\eq0\bmod p$, then there exists a unique $t\bmod p$ such that $f(a+tp^j)\eq0\bmod p^{j+1}$}.
\end{Lemma}

\begin{Def}
\textit{The polynomial $g(X)$ is a \textbf{partial $\Yb$-root of precision $b$} in $Q(X,\Yb)$, if $g(X)\eq f(X)(\bmod X^b)$, for some $\Yb$-root $f(X)\in\F_q[X]$ of $Q(X,\Yb)$}.
\end{Def}

$\ind$ We note that the degree of a partial $\Yb$-root of precision $b$, is at most $b-1$, and Hensel-lifting is a general procedure for computing such roots. The key is to recursively \textit{lift} a partial root of degree $b$ to a new one of precision $b+1$, as in the simplest case of Hensel's lemma \ref{Hens_lem_simple}. Let $Q$ satisfy (\ref{int_cond_frs}) (in \cite{Bra10} such polynomials are referred to as \textit{interpolation polynomials}). For our purposes, we may consider the nonzero polynomial (\ref{Q_interp}). From (\ref{s_list_cond}), if $f(X)=\sum_{i=0}^{k-1}f_iX^i$ is a $\Yb$-root of $Q$, then
\begin{equation}
\label{cond_b_1} 
  Q(X,f(X),f(\gamma X),\cdots,f(\gamma^{s-1} X)) \ \equiv \ Q(0,f_0,\cdots,f_0)\left(\bmod X\right) \ \eq \ 0 \left(\bmod X\right).
\end{equation}
Since $X\nmid Q(0,f_0,\cdots,f_0)$, it follows that $Q(0,f_0,\cdots,f_0)=0$. A partial root $f_0$ of precision $b=1$ must therefore satisfy this condition.\\
$\ind$ It is clear that if $X^r\mid Q(X,\Yb)$ for some $r\in\Z_+$, then any $\Yb$-root of $Q$ will also be a root of $X^{-r}Q$ --- this idea resembles the constructive proof of lemma \ref{Hens_lem_simple}. We may therefore assume that $X\nmid Q$, or equivalently that $Q(0,\Yb)\not\eq0$. There is a subtlety here though.  If $Q(0,\Yb)\in\I$, for $\I$ the proper ideal $\I=\langle Y_1-Y_2,Y_2-Y_3,\cdots,Y_{s-1}-Y_s\rangle\lhd \F_q[\Yb]$; then $Q(0,Y,\cdots,Y)=0$, and (\ref{cond_b_1}) reveals nothing about $f_0$. In such a case, we have $q$ partial roots of precision $b=1$. Furthermore if $Q(0,Y,\cdots,Y)\not\eq0$, by (\ref{cond_b_1}) $f_0$ must be among the roots of this polynomial. Since deg$(Q(X,f(X),\cdots,f(\gamma^{s-1}X)))\leq$deg$(Q(0,\Yb))$, one can constrain the number of possible partial roots of  precision $b=1$ (at most $\ell=\floor{\frac{\Delta-1}{k-1}}$, for $\Delta$ defined in \cite{Bra10} corollary 5.6).\\
$\ind$ For lifting partial roots, under the assumption that $f(X)=f_0+X\tilde{f}(X)$ is a $\Yb$-root of $Q$, i.e.
$$ Q(X,f(X),f(\gamma X),\cdots,f(\gamma^{s-1} X)) = Q(X,f_0+X\tilde{f}(X),f_0+\gamma X\tilde{f}(\gamma X),\cdots,f_0+\gamma^{s-1} X\tilde{f}(\gamma^{s-1} X))=0 $$
and $f_0$ is known, it follows that $\tilde{f}(X)$ is a $\Yb$-root of
\begin{equation*}
  \tilde{Q}(X,\Yb)\coloneqq Q(X,f_0+X\cdot Y_1,f_0+\gamma X\cdot Y_2,\cdots,f_0+\gamma^{s-1} X\cdot Y_s)
\end{equation*}
where now $\tilde{f}(\gamma^{i-1} X)$ from the above identity ``replace'' the corresponding variable $Y_i$, for all $i\in\N_s$.

\begin{Lemma}
\label{lifting_lemma}
\textit{Let the polynomial $Q(X,\Yb)\in\F_q[X,\Yb]$ be nonzero, and $b\in \F_q$. Then the polynomial $Q(X,b+X\cdot Y_1,\cdots, b+\gamma^{s-1}X\cdot Y_s)$ is also nonzero}.
\end{Lemma}

\begin{proof}
Define the bijection $\phi_{b}:\F_q[X][\Yb]\to\F_q[X][\Yb]$ as
$$ \phi_{b}\big(P(X,Y_1,Y_2,\cdots,Y_s)\big) = P(X,b+X\cdot Y_1,b+\gamma X\cdot Y_2,\cdots,b+\gamma^{s-1} X\cdot Y_s) $$
for $P(X,\Yb)\in\F_q[X][\Yb]$. If $P=0$, then $\phi_{b}(P)=0$. Since $\F_q[X,\Yb]\cong\F_q[X][\Yb]$, we may assume that $P(X,\Yb)\in\F_q[X,\Yb]$. By assuming that $Q(X,\Yb)\not\eq0$, it follows that $\phi_{b}(Q)\not\eq0$.
\end{proof}

$\ind$ This lemma is precisely what we need for lifting the partial roots, and can be viewed as another case of lemma \ref{Hens_lem_simple}. Under the assumption that $Q(X,\Yb)\not\eq0$, it follows that $\tilde{Q}(X,\Yb)\not\eq0$. By lemma \ref{lifting_lemma} it follows that there exists an integer $r\geq0$ for which $X^r\parallel \tilde{Q}(X,\Yb)$, and we define
$$ Q_{f_0}(X,\Yb)=X^{-r}\cdot\tilde{Q}(X,\Yb). $$
Recall that $f_0$ was defined such that $f(X)=f_0+X\tilde{f}(X)$ is a $\Yb$-root of $Q$, so by evaluating $Q(X,\Yb)$ at $X=0$ we get $Q(0,\Yb)=Q(0,f_0,\cdots,f_0)=0$. This implies that $X\mid Q(X,\Yb)$, hence $r$ is positive. Since $\tilde{f}(X)$ is a $\Yb$-root of $\tilde{Q}$, it follows that it is also a $\Yb$-root of $Q_{f_0}$.\\
$\ind$ The above discussion can be summarized in the recursive expression
\begin{equation}
\label{rec_expr}
  \Phi_k(Q) \subseteq \bigcup_{f_0\in\Phi_1(Q)} \big(f_0+X\cdot\Phi_{k-1}(Q_{f_0})\big)
\end{equation}
where $\Phi_k(Q)$ denotes the set of partial $\Yb$-roots of $Q$, of precision $k$. In the recursive expression, it is clear that partial roots of precision $b$, are also partial roots of precision $b+1$. By definition, $\Phi_k(Q)$ contains the set of all $\Yb$-roots of $Q$ of degree at degree at most $k-1$. Likewise, we define the \textit{list} of polynomials
$$ \Lambda_{k}(Q) = \bigcup_{f_0\in\Lambda_1(Q)} \big(f_0+X\cdot\Lambda_{k-1}(Q_{f_0})\big) \qquad \text{where} \qquad \Lambda_1(Q)=\big\{f_0\in\F_q \mid Q(0,f_0,\cdots,f_0)=0\big\} $$
and by the condition we showed that precision $b=1$ partial roots $f_0$ must satisfy (\ref{cond_b_1}), we have
$$ \Phi_k(Q) \subseteq \Lambda_{k}(Q) $$
for all $k\in\Z_+$. We can attain $\Lambda_1(Q)$ by enumerating all $f_0\in\F_q$ satisfying (\ref{cond_b_1}), hence; we can recursively compute $\Lambda_k(Q)$ by algorithm \ref{enum_L_k_alg}.

\begin{algorithm}[H]
\label{enum_L_k_alg}
\SetAlgoLined
\KwIn{$k\in\Z_+$ and $Q(X,\Yb)\in\F_q[X,Y]$}
\KwOut{$\Lambda_k(Q)$ --- a list containing all precision $k$ $\Yb$-roots of $Q$}
  \eIf {$k\leq0$}
    {$\Lambda_k(Q)\gets\{0\}$}
    {let $Q_{f_0}(X,\Yb)\gets X^{-r}\cdot Q(X,\Yb)$, for $r$ s.t. $X^r\parallel Q(X,\Yb)$\\
    \eIf {$Q_{f_0}(0,Y,\cdots,Y)=0$}
      {$B\gets \F_q$}
      {$B\gets\big\{\beta\in\F_q : Q(0,\beta,\cdots,\beta)=0\big\}$}
    }
  \Return $\Lambda_k(Q)\gets\bigcup\limits_{f_0\in B}\big(f_0+X\cdot\Lambda_{k-1}(Q_{f_0})\big)$ \Comment{compute recursively}
\caption{Enumeration of $\Lambda_k(Q)$}
\end{algorithm}

$\ind$ By our previous discussions, the set $\Lambda_k(Q)$ contains the list of polynomials we are looking for, from \cref{root_find_subsec}. This list can then be reduced to the set of $\Yb$-roots of $Q$ of degree at most $k-1$, by retaining only the polynomials which satisfy (\ref{s_list_cond}). For further comparison of the root-methods, refer to \cite{Bra10} sections 5.4.3.

\section{Locally Decodable Multiplicity Codes}
\label{LDCs_Mult_codes_sec}

We now shift gears and turn our attention to \textit{multiplicity codes} \cite{KSY10}, a type of locally decodable error-correcting codes (LDCs). Recall that the main parameters of LDCs are its length $n$ (or its rate $R=k/n$ for fixed $k$) and \textit{query complexity} of local decoding. Ideally, we would like to have both of these parameters be small, though one cannot minimize them both simultaneously.\\
$\ind$ Most work prior to \cite{KSY10} focused on studying codes in the low and constant query regimes, which have applications in cryptography and complexity theory. Multiplicity codes on the other hand, were introduced in order to study how the query complexity of large rate (approaching 1) LDCs can be minimized. Before the construction of these codes, it was unknown how to get any nontrivial local decoding for codes of rate $R>1/2$.\\
$\ind$ This relatively new family of codes has been named multiplicity codes, as they are based on evaluating multivariate polynomials and their derivatives, while also considering high-multiplicity zeroes. By the way they are defined, they inherit the local decodability of the classical multivariate polynomial codes based, while achieving better trade-offs and flexibility in the rate and minimum distance. In \cref{Der_Codes_sec} we will see how variants of these codes relate to FRS codes.\\
$\ind$ Before we start, let us define a \textit{local self-correction} property, as for our purposes we want to construct ``locally self-correctable codes'' (LSCCs) over ``large alphabets'' $\Sigma$. The code $\Cc\subseteq\Sigma^n$ of size $|\Cc|=|\Sigma|^k$ and large $R$ we want to construct, should satisfy the property: given access to a received string $r\in\Sigma^n$ which is close to some $c\in\Cc$, and given any coordinate index $i\in\N_n$, it is possible to make few queries to the coordinates of $r$, and with high probability retrieve $c_i$. We point out that this is different from the notion of locally decodability, where to goal is to recover the coordinate of the original message $\vec{m}\in\Sigma^k$. We show in \cref{Mult_codes_subs} though that for linear codes, LSCCs imply LDCs. Throughout \cref{LDCs_Mult_codes_sec}, $q$ denotes a power of a prime $p$.

\subsection{Bivariate Multiplicity Codes}
\label{Biv_Mult_codes_subs}

In order to define the bivariate multiplicity codes; the simplest example of multiplicity codes, we first need to give several definitions. The bivariate multiplicity codes already have improvements the in terms of rate for local self-correction over Reed-Muller codes (RM), while being locally self-correctable with only a constant factor more queries.\\
$\ind$ For a vector $\ib=(i_1,\cdots,i_n)\in\N_0^{n}$, we denote its \textit{weight} by wt$(\ib)=\|\ib\|_1=\sum_{j=1}^n\ib_j$. As in \cref{int_step_subsec_frs}, denote the formal variables $X_1,\cdots,X_n$ by $\Xb$, thus $\F[\Xb]=\F[X_1,\cdots,X_n]$. Lastly, for $\ib\in\N_0^{n}$ let $\Xb^{\ib}$ denote the monomial $\prod_{j=1}^nX_{j}^{\ib_j}\in\F[\Xb]$. It follows that (total) $\text{deg}\left(\Xb^{\ib}\right)=\text{wt}(\ib)$.

\begin{Def}
\label{Hasse_der}
\textit{For $P(\Xb)\in\F[\Xb]$ and $\ib\in\N_0^n$, the \textbf{$\ib^{th}$ Hasse derivative of $P$}, denoted $P^{(\ib)}(\Xb)\in\F[\Xb]$, is the coefficient of $\Zb^{\ib}$} \textit{in the polynomial} $\tilde{P}(\Xb,\Zb)\coloneqq P(\Xb+\Zb)\in\F[\Xb,\Zb]$. \textit{Thus}
$$ P(\Xb+\Zb)=\sum_{\ib}P^{(\ib)}(\Xb)\Zb^{\ib} $$
\textit{and observe that for all} $P,Q\in\F[\Xb]$ \textit{and} $\lambda\in\F$
$$ (\lambda P)^{(\ib)}(\Xb)=\lambda P^{(\ib)}(\Xb) \qquad \textit{ and } \qquad P^{(\ib)}(\Xb)+Q^{(\ib)}(\Xb) = (P+Q)^{(\ib)}(\Xb). $$
\end{Def}

\begin{Def}
\label{mult_def}
\textit{For $P(\Xb)\in\F[\Xb]$ and $\ab\in\F^n$, the \textbf{multiplicity} of $P$ at $\ab$, denoted by} mult$(P,\ab)$, \textit{is the largest integer $M$ such that for every non-negative vector $\ib$ with} wt$(\ib)<M$, \textit{we have $P^{(\ib)}(\ab)=0$ (if $M$ is taken arbitrarily large, we set} mult$(P,\ab)=\infty$\textit{). Note that} mult$(P,\ab)\geq0$ \textit{for every $\ab$}. 
\end{Def}

\begin{Def}
\label{biv_mult_code}
\textit{The multiplicity code of \textbf{order 2} evaluations of degree $d=2(1-\delta)q$ bivariate polynomials over $\F_q$ for $\delta>0$, is the set of codeword vectors corresponding to the polynomials $P(X,Y)\in\F_q[X,Y]$}
$$ C(P)=\left\langle\left(P(\ab),\frac{\partial P}{\partial X}(\ab),\frac{\partial P}{\partial Y}(\ab)\right)\right\rangle_{\ab\in\F_q^2}\in(\F_q^3)^{q^2}\cong\F_{q^3}^{q^2} $$
\textit{where $C(P)$ indicates the encoding of $P$. The coordinates are indexed by $\F_q^2$; thus $n=q^2$, and the codewords are indexed by the bivariate polynomials of degree at most $d$ over $\F_q$}.
\end{Def}

$\ind$ In simpler words, the $\ab$ coordinate consists of the evaluations of $P$ and its two partial derivatives at $\ab$. By \cite{DKSS13} lemma 8 (strengthening of the Schwartz-Zippel lemma), it follows that two distinct polynomials of degree at most $d$ can agree \underline{with multiplicity 2} on at most $d/2q$-fraction of the points in $\F_q^2$, hence this codes has relative distance $\delta=1-d/2q$. Since now $|\Sigma|=q^3$, the message length $k$ equals the number of $q^3$-ary symbols required to specify a polynomial of degree at most $d$. Since we have $d+1$ monomials, 2 variables and are ``grouping'' the elements in pairs of three (going from $q$-ary to $q^3$-ary), we get $k={d+1 \choose 2}/3$. The rate of the bivariate multiplicity code is therefore 
$$ R=\frac{k}{n} = \frac{{d+1 \choose 2}/3}{q^2} \simeq \frac{\frac{d^2}{2}/3}{q^2} = \frac{\frac{\big(2q(1-\delta)\big)^2}{2}/3}{q^2} = \frac{2(1-\delta)^2}{3} \qquad \Longrightarrow \qquad \lim_{\delta\to0}\left\{R\right\}=\frac{2}{3} $$
an improvement to the rate of the corresponding RM code; which was less than $\frac{1}{2}$, while having the same distance. The bivariate RM code is instead defined by $C(P)=\left\langle P(\ab)\right\rangle_{\ab\in\F_q^2}\in\F_{q}^{q^2}$, and has parameters $\delta=1-d/q$, $k={d+1 \choose 2}$, $n=q^2$, thus $R={d+1 \choose 2}/q^2\simeq \frac{(1-\delta)^2}{2}<\frac{1}{2}$.

\subsection{Local Self-Correction of Bivariate Multiplicity Codes}
\label{LSC_Biv_Mult_codes_subs}

We now see how local self-correction is achieved. Given a received word $r\in(\F_q^3)^{q^2}$ close to the codeword $C(P)$ in terms of Hamming distance $\Delta_H(\cdot,\cdot)$, we want to recover the ``correct'' symbol at coordinate $\ab$ of a given point $\ab\in\F_q^2$, namely $\left(P(\ab),\frac{\partial P}{\partial X}(\ab),\frac{\partial P}{\partial Y}(\ab)\right)$. The approach is similar to local self-correction of RM codes, where we pick a random direction $\bb\in\F_q^2$ and look at the restriction of $r$ to coordinates in the line $L=\{\ab+\bb t\mid t\in\F_q\}$. With high probability over the choice of $\bb$, $r|_L$ and $C(P)|_L$ agree in many locations; i.e. $\Delta_H\left(r|_L,C(P)|_L\right)$ is small. The next step is to recover $Q(T)=P(\ab+\bb T)$ for which deg$(Q)\leq 2(1-\delta)q$, in order to compute the $3$-tuple defining $C(P)$.\\
$\ind$ It is important to notice that for every $t\in\F_q$, the $\ab+\bb t\in L$ coordinate of $C(P)$ completely determines both the value and the $1^{st}$ derivative of $Q(T)$ at point $t$, as by the chain rule we have
$$ \left(Q(t),\frac{\partial Q}{\partial T}(t)\right) = \left(P(\ab+\bb t),\bb_1\frac{\partial P}{\partial X}(\ab+\bb t)+\bb_2\frac{\partial P}{\partial Y}(\ab+\bb t)\right) \in \F_q^2. $$
Our knowledge of $r|_L$ therefore gives us access to $q$ ``noisy'' evaluations of $Q(T)$ (one for each $t\in\F_q$), and its derivative $\frac{\partial Q}{\partial T}(T)$, which is enough for recovering $Q(T)$. Clearly $Q(0)=P(\ab)$, and $\frac{\partial Q}{\partial T}(0)=\bb_1\frac{\partial P}{\partial X}(\ab)+\bb_2\frac{\partial P}{\partial Y}(\ab)$ is the directional derivative of $P$ at $\ab$ in direction $\bb$.\\
$\ind$ We repeat the above for a different direction $\grave{\bb}\in\F_q^2$ and $\grave{Q}(T)=P(\ab+\grave{\bb}t)$, to recover the directional derivative $\frac{\partial \grave{Q}}{\partial T}(0)=\grave{\bb}_1\frac{\partial P}{\partial X}(\ab)+\grave{\bb}_2\frac{\partial P}{\partial Y}(\ab)$ of $P$ at $\ab$ in direction $\grave{\bb}$. Together, the two directional derivatives $\frac{\partial Q}{T}(0)$ and $\frac{\partial \grave{Q}}{\partial T}(0)$ suffice to recover $\frac{\partial P}{\partial X}(\ab)$ and $\frac{\partial P}{\partial Y}(\ab)$, as we have a linear system of two equations; with two unknowns which we want to recover. All in all, this approach makes $2q=O(\sqrt{k})$ queries; needed for the ``noisy'' evaluations of $Q(T)$ and $\grave{Q}(T)$. This sublinear query complexity was something not known before, for local decoding in the regime of $R>1/2$.

\subsection{Multiplicity Codes and Local Self-Correction}
\label{Mult_codes_subs}

In order to get multiplicity codes of rate approaching $1$, we also consider evaluations of all derivatives of the multivariate polynomial $P$ up to an even higher order. To locally recover the evaluations of the higher order at a point $\ab$, we pick many random lines passing through $\ab$, try to recover the restriction of $P$ to those lines (correspond to univariate polynomials), which we combine in a certain way. The procedure is formally explained in algorithm \ref{simpl_LSC_alg}. By simultaneously increasing the maximum order of derivative taken and the number of variables, we attain multiplicity codes with the desired rate and local decodability. To state the results on the existence of LDCs with rate approaching 1, we need the following definitions.

\begin{Def}
\textit{The \textbf{relative Hamming distance} of two strings $c,c'\in\Sigma^n$, is the fraction of coordinates in which they differ: $\delta_H(c,c')=\frac{\Delta_H(c,c')}{n}=\Pr_{i\in\N_n}\left[c_i\neq c_i'\right]$}.
\end{Def}

\begin{Def}[\textbf{Locally Self-Correctable Code}]
\textit{A code $\Cc\subseteq\Sigma^n$ is said to be \textbf{locally self-correctable} from $\delta'$-fraction errors with $t$ queries, if there is a randomized algorithm} $\Alg$ \textit{such that}:
\begin{itemize}
  \item \textbf{Self-Correction}: \textit{Whenever $\delta_H(r,c)<\delta'$ for $c\in\Cc$ and $r\in\Sigma^n$, then for each $i\in\N_n$}
  $$ \Pr\left[\Alg^r(i)=c_i\right]\geq2/3 $$
  \item \textbf{Query Complexity $t$}: $\Alg^r(i)$ \textit{always makes at most $t$ queries to $r$}
\end{itemize}
\textit{where} $\Alg^r$ \textit{represents the situation where $\Alg$ is given query access to $r$}.
\end{Def}

\begin{Def}[\textbf{Locally Decodable Code}]
\textit{Let $\Cc\subseteq\Sigma^n$ be a code with $|\Cc|=|\Sigma|^k$, and $E:\Sigma^k\to\Cc$ a bijection; which is $\Cc$'s encoding map. We say that $(\Cc,E)$ is \textbf{locally decodable} from $\delta'$-fraction errors with $t$ queries, if there is a randomized algorithm} $\Alg$ \textit{such that}:
\begin{itemize}
  \item \textbf{Decoding}: \textit{Whenever $\vec{m}\in\Sigma^k$ and $r\in\Sigma^n$ are such that $\delta_H(r,E(\vec{m}))<\delta'$, then for each $i\in\N_k$}
  $$ \Pr\left[\Alg^r(i)=m_i\right]\geq2/3 $$
  \item \textbf{Query Complexity $t$}: $\Alg^r(i)$ \textit{always makes at most $t$ queries to $r$}
\end{itemize}
\textit{where} $\Alg^r$ \textit{represents the situation where $\Alg$ is given query access to $r$}.
\end{Def}

$\ind$ Recall that any linear code has a systematic encoding, which means there is an encoding $E$ such that for each $\vec{m}\in\Sigma^k$ and $i\in\N_k$, there is a $j\in\N_n$ such that $E(\vec{m})_j=m_i$. This gives us the implication that if $\Cc$ is a LSCC, then $(\Cc,E)$ is a LDC, with the same fraction of errors $\delta'$ and query complexity $t$. We can view this implication as a reduction, which allows us to focus on constructing linear LSCCs. There is a caveat here, the fact that multiplicity codes are not linear codes. However, it is possibly to achieve linear LSCCs by concatenating multiplicity codes with suitable ``good'' linear codes over the small alphabet $\Sigma=\F_p$. The resulting LDCs have similar parameters. Furthermore, multiplicity codes themselves can also be locally decoded with a factor exp$(m+s)$-increase in the query complexity, for a suitable encoding $E$. Though obvious, it is also important to point out that local decoding is a function of the encoding $E$.\\
$\ind$ We now define multiplicity codes, state and prove their rate and distance, and then show how their local self-correction is achieved. The compelling part about the relationship between rate and distance, is that if we keep $\delta$ fixed and let the multiplicity parameter $s$ grow the rate improves, as it approaches $(1-\delta)^m$. For our constructions, we assume that $\Sigma=\F_q$. 

\begin{Def}
\label{mult_codes_def}
\textit{Let $s,d,m\in\N_0$ and} $\Sigma=\F_q^{{m+s-1 \choose m}}=\F_q^{|\{\ib:\text{wt}(\ib)<s\}|}$. \textit{For $P(\Xb)\in\F_q[\Xb]$ where $\Xb=(X_1,\cdots,X_m)$, and $\ab\in\F_q^m$, the \textbf{order $s$ evaluation of $P$ at $\ab$}, denoted $P^{(<s)}(\ab)$, is the vector} $\langle P^{(\ib)}(\ab)\rangle_{\text{wt}(\ib)<s}\in\Sigma$. \textit{The \textbf{multiplicity code of order $s$ evaluations of degree $d$ polynomials in $m$ variables over $\F_q$}, is the code over $\Sigma$ of length $n=q^m$ (where the coordinates are indexed by the elements $\ab\in\F_q^m$). For each $P(\Xb)\in\F_q[\Xb]$ with} deg$(P)\leq d$, \textit{there is a codeword in $\Cc$ given by the encoding}:
$$ \Enc_{s,d,m,q}(P)\coloneqq \left\langle P^{(<s)}(\ab)\right\rangle_{\ab\in\F_q^m}\in(\Sigma)^{q^m}. $$
\end{Def}

\begin{Lemma}[Rate and distance of multiplicity codes]
\textit{Let $\Cc$ be a multiplicity code of order $s$ evaluations of degree $d$ polynomials in $m$ variables over $\F_q$. Then $\Cc$ has $\delta=1-\frac{d}{sq}$ and $R=\frac{{d+m \choose m}}{{s+m-1 \choose m}q^m}$, for which $R\geq \left(\frac{s}{m+s}\right)^m\cdot\left(\frac{d}{sq}\right)^m\geq\left(1-\frac{m^2}{s}\right)\left(1-\delta\right)^m$}.
\end{Lemma}

\begin{proof}
Consider two codewords $c_1=\Enc_{s,d,m,q}(P_1)$ and $c_2=\Enc_{s,d,m,q}(P_2)$ where $P_1\neq P_2$. For the coordinates $\ab\in\F_q^m$ where $c_1,c_2$ agree; i.e. $(c_1)_{\ab}=(c_2)_{\ab}$, we have $P_1^{(<s)}(\ab)=P_2^{(<s)}(\ab)$. Consequently, for any such $\ab$ we have $(P_1-P_2)^{(\ib)}(\ab)=0$ for each $\ib\in\{\ib:\text{wt}(\ib)<s\}$, thus mult$(P_1-P_2,\ab)\geq s$. From \cite{KSY10},\cite{DKSS13} lemmas 7 and 8 respectively, mult$(P_1-P_2,\ab)\geq s$ can occur  on a fraction of at most $\frac{d}{sq}$ points $\ab\in\F_q^m$. The minimum relative distance $\delta$ of $\Cc$ is therefore at least $\delta\geq 1-\frac{d}{sq}$.\\
$\ind$ We now compute the code's rate $R=\frac{\log|\Cc|}{n\log|\Sigma|}$. By definition \ref{mult_codes_def} our alphabet size is $q^{{m+s-1 \choose m}}$ and block-length is $n=q^m$, so it remains to calculate $|\Cc|$. A codeword is specified by giving coefficients to each of the monomials of degree at most $d$, thus $|\Cc|=q^{d+m \choose m}$. The rate is therefore
$$ R = \frac{{d+m \choose m}}{{s+m-1 \choose m}q^m} = \frac{\prod_{j=0}^{m-1}(d+m-j)}{\prod_{j=1}^{m}\big((s+m-j)q\big)} \geq \left(\frac{1}{1+\frac{m}{s}}\right)^m\cdot\left(\frac{d}{sq}\right)^m \geq \left(1-\frac{m^2}{s}\right)\left(1-\delta\right)^m. $$
\end{proof}

$\ind$ Using the parameters of definition \ref{mult_codes_def}, let $r:\F_q^m\to\Sigma$ be a received word for $\Sigma$ our code's alphabet. Suppose $P(\Xb)\in\F_q[\Xb]$ has deg$(P)\leq d$ such that $\delta_H\left(r,\Enc_{s,d,m,q}(P)\right)$ is small, and let $\ab\in\F_q^m$. Before showing how to locally recover $P^{(<s)}(\ab)$ (algorithm \ref{simpl_LSC_alg}) when given oracle access to $r$, we establish two relationships between the derivatives of the restriction of $P$ to a line to the derivatives of $P$ itself. Fix $\ab,\bb\in\F_q^m$ for $\bb\neq0$, and consider the polynomial $Q(T)=P(\ab+\bb T)$.
\begin{itemize}
  \item \textbf{Relationship of $Q(T)$ with the derivatives of $P$ at $\ab$}: By \ref{Hasse_der}: $Q(T)=\sum\limits_{\ib}P^{(\ib)}(\ab)\bb^{\ib}T^{\text{wt}(\ib)}$. By grouping terms: $\left(\sum\limits_{\ib|\text{wt}(\ib)=e}P^{(\ib)}(\ab)\bb^{\ib}\right)=$ coefficient of the monomial $T^e$ in $Q(T)$.
  \item \textbf{Relationship of derivatives of $Q(T)$ at $t$ with the derivatives of $P$ at $\ab+\bb t$}: By \ref{Hasse_der}: $P\big(\ab+\bb(t+R)\big)=Q(t+R)=\sum\limits_{j}Q^{(j)}(t)R^j$ and $P\big(\ab+\bb(t+R)\big)=\sum\limits_{\ib}P^{(\ib)}(\ab+\bb t)(\bb R)^{\ib}$, for $t\in\F_q$. Thus: $Q^{(j)}(t)=\left(\sum\limits_{\ib|\text{wt}(\ib)=j}P^{(\ib)}(\ab+\bb t)\bb^{\ib}\right)$. More precisely, $Q^{(j)}(t)$ is a linear combination of various $P^{(\ib)}(\ab+\bb t)$ with coefficients $\bb^\ib$, over different $\ib$ of weight $j$.
\end{itemize}

Recall that we want to recover $P^{(<s)}(\ab)$, where $P(\Xb)$ is such that $\Enc_{s,d,m,q}(P)$ is close to $r$, i.e. $\delta_H\left(r,\Enc_{s,d,m,q}(P)\right)$ is small. We denote the $\ib$ coordinate of $r(\ab)$ by $r^{(\ib)}(\ab)$. This is done by algorithm \ref{simpl_LSC_alg}.

\begin{algorithm}[H]
\label{simpl_LSC_alg} 
\SetAlgoLined
\KwIn{received word $r:\F_q^m\to\Sigma$, and point $\ab\in\F_q^m$}
\KwOut{Vector $\langle u_\ib\rangle_{\text{wt}(\ib)<s}$}
  \begin{enumerate}
    \item \textbf{Pick a set of directions $B$}: Choose $B\subseteq\F_q^m\backslash\{\bold{0}\}$ uniformly at random, of size ${m+s-1 \choose m}$.
    \item \textbf{Recover $P(\ab+\bb T)$ for $\bb\in B$}: For each $\bb\in B$ consider $\ell_\bb:\F_q\to\F_q^s$ given by
    $$ \big(\ell_\bb(t)\big)_j = \sum\limits_{\ib|\text{wt}(\ib)=j}r^{(\ib)}(\ab+\bb t)\bb^{\ib}. $$
    Find $Q_\bb(T)\in\F_q[T]$ with deg$(Q_\bb)\leq d$ (if any), s.t. $\delta_H(\Enc_{s,d,1,q}(Q_\bb),\ell_\bb)<\delta/2$.
    \item \textbf{Solve a linear system to recover $P^{(<s)}(\ab)$}: For each $e\in\{0,1,\cdots,s-1\}$ consider the system of equations in the variables $\langle u_\ib\rangle_{\text{wt}(\ib)=e}$ (with one equation for each $\bb\in B$):
    \begin{equation}
    \label{step_3_lin_syst}
      \left(\sum\limits_{\ib|\text{wt}(\ib)=e}\bb^{\ib}u_\ib\right)= \text{ coefficient of } T^e \text{ in } Q_\bb(T).
    \end{equation}
    Find all $\langle u_\ib\rangle_{\text{wt}(\ib)=e}$ satisfying the system (\ref{step_3_lin_syst}).
    \item \textbf{Existence and uniqueness}:
    If the solution does \textit{not} exist or is \textit{not} unique, output \textbf{FAIL}.
  \end{enumerate}
\caption{Simplified Local Self-Correction of Multiplicity Codes}
\end{algorithm}

\begin{Claim}
\textit{Algorithm \ref{simpl_LSC_alg} is a local-self corrector from a $\frac{\delta}{100w}$-fraction of errors, for $w={m+s-1 \choose m}$. Overall, it outputs $P^{(\ib)}(\ab)$ with probability at least $\left(\frac{9}{10}\right)^2\simeq 0.8$}.
\end{Claim}

$\ind$ We validate the above claim, by analyzing the three steps of the algorithm. Fix a received word $r:\F_q^m\to\Sigma$ and $\ab\in\F_q^m$, and let $P(\Xb)$ be a polynomial such that $\delta_H\big(\Enc_{s,d,m,q}(P),r\big)<\frac{\delta}{100w}$. We call the points where $r$ and $\Enc_{s,d,m,q}(P)$ differ the ``errors''.\\

\underline{\textbf{Step 1 -- All $\bb\in B$ are ``good''}}: For a fixed $\bb\in\F_q^m\backslash\{\bold{0}\}$, we are interested in the fraction of errors on the line $L_{\bb}=\{\ab+\bb t\mid t\in\F_q^{\times}\}$ through $\ab$ in direction $\bb$. Considering the space $\F_q^m$, the lines defined from the points in $B$ cover $\F_q^m\backslash\{\ab\}$ uniformly. By this, at most $\frac{1}{50w}$ of the lines containing $\ab$ have more than a $\frac{\delta}{2}$-fraction error on them. Therefore
$$ \Pr_{B\gets\mathcal{P}_w\left(\F_q^m\backslash\{\bold{0}\}\right)}\left[\text{all $\bb\in B$ will be s.t. $L_{\bb}$ through $\ab$ has $\frac{\delta}{2}$ errors on it}\right]\geq \frac{9}{10} $$
where $\mathcal{P}_w$ denotes the subsets of cardinality $w$.

\underline{\textbf{Step 2 -- $Q_\bb(T)=P(\ab+\bb T)$ for each $\bb\in B$}}: In the case where $B$ satisfies the above event, by the third equation of the relationship relating the derivatives of $Q$ and $P$ (identity for $Q^{(j)}(t)$), for each $\bb$, the corresponding $\ell_\bb$ will satisfy $\delta_H\big(\Enc_{s,d,1,q}(P(\ab+\bb T)),\ell_\bb\big)<\delta/2$. Thus, for each $\bb\in B$, the algorithm will find $Q_{\bb}(T)=P(\ab+\bb T)$.\\

\underline{\textbf{Step 3 -- $u_\ib=P^{(\ib)}(\ab)$ for each $\ib$}}: Since $Q_\bb(T)=P(\ab+\bb T)$ for each $\bb\in B$, by the second identity of our first relationship, we get that for each $e\in\{0,1,\cdots,s-1\}$ the vector $\langle u_\ib\rangle_{\text{wt}(\ib)=e}$ with $u_\ib=P^{(\ib)}(\ab)$ will satisfy all the equations in the system (\ref{step_3_lin_syst}), which solution is unique. Furthermore
$$ \Pr_{B\gets\mathcal{P}_w\left(\F_q^m\backslash\{\bold{0}\}\right)}\Big[\text{the elements of $B$ form an \underline{interpolating set} for polynomials of degree $<s$}\Big]\geq\frac{9}{10} $$
which holds as long as $q$ is large enough in terms of $m$ and $s$. In particular, no $P(\Xb)\in\F_q[\Xb]\backslash\{0\}$ of deg$(P)<s$ vanishes on all $\bb\in B$. If the solution to (\ref{step_3_lin_syst}) was not unique and had distinct solutions $u_\ib$ and $u_\ib'$, then $\langle u_\ib-u_\ib'\rangle_{\text{wt}(\ib)=e}$ would be the vector of coefficients of a polynomial $\tilde{P}(\Xb)\in\F_q[\Xb]\backslash\{0\}$ of deg$(\tilde{P})<s$ which vanishes on all $\bb\in B$. Therefore
$$ \left(\sum\limits_{\ib|\text{wt}(\ib)=e}\tilde{P}^{(\ib)}(\ab)\bb^{\ib}\right) = \left(\sum\limits_{\ib|\text{wt}(\ib)=e}\left(u_\ib-u_\ib'\right)\bb^{\ib}\right) = 0 $$
which contradicts the fact that $B$ is an \textbf{interpolating set} for polynomials of degree $<s$ (a subset of $\F_q^m$, for which if we are given $\{q(\bb)\}_{\bb\in B}$ for $q(\Xb)\in\F_q[\Xb]$, we can reconstruct $q(\Xb)$ \cite{DS08}).\\

$\ind$ A central result of \cite{KSY10} (theorem 10) states that for $q\geq \max\{10m,\frac{d+6}{s},5(s+1)\}$ and $\delta=1-\frac{d}{sq}$, $\Cc$ is locally self-correctable from $\frac{\delta}{10}$-fraction errors with $q\cdot O(s)^m$ queries. The proof of the theorem in which this statement appears uses a slightly different algorithm for local self-correction of multiplicity codes. For more details and further results on multiplicity codes, please refer to \cite{KSY10}, \cite{Kop13}, and \cite{Kop15}; in which another explicit capacity-achieving list decodable code was developed.

\section{Derivative Codes}
\label{Der_Codes_sec}

We present one last family of codes, \textit{derivative codes} \cite{GW11},\cite{GW13}, which we relate to the others presented thus far. This gives an alternate construction to FRS codes, for achieving the optimal trade-off between rate and list decoding error-correction radius. Informally, rather than bundling evaluations of the message polynomial $f(X)$ at consecutive powers of $\gamma$ as in (\ref{frs_eq}), in an order-$m$ derivative code, we bundle the evaluations of $f(X)$ along with its first $(m-1)$ derivatives at each point of the defining set of points $\A=\{\alpha_1,\cdots,\alpha_n\}\subseteq\F_q$. This resemblance makes this construction arguably just as natural as that of FRS codes. An interesting artifact of this construction is that the rate does not decrease, as one can pick higher degree polynomials; while still maintaining the distance. The reason is that two distinct polynomials of degree $\ell$ and their first $(m-1)$ derivatives, can agree in at most $\ell/m$ points.\\
$\ind$ The list decoding of derivative codes involves an interpolation step, and a second step of retrieving the list of polynomials satisfying a certain algebraic condition, similar to what we saw for FRS codes. The first step consists of fitting a polynomial of the form (\ref{Q_interp}). The second step which is new to us, consists of solving a ``differential equation''. This was also considered in \cite{Kop15}, where the power series expansion of the potential solution was used to solve the same differential equation. Without further ado, let us define derivative codes.

\begin{Def}
\label{Der_codes_def}
\textit{Let $m\in\N_0$ and $\alpha_1,\cdots,\alpha_n\in\F_q$ be distinct, and the parameters satisfy $m\leq k<nm\leq q$. Further assume that} char$(\F_q)>k$. \textit{The \textbf{$m^{th}$ order derivative code}} $\Der_q^{(m)}[n,k]$ \textit{over the alphabet $\F_q^m\cong\F_{q^m}$, encodes the polynomial $f(X)\in\F_q[X]$ with} deg$(f)=k-1$ \textit{by}
\begin{equation}
\label{Der_code_eq}
  {\footnotesize f(X) \mapsto \left( \begin{bmatrix} f(\alpha_1) \\ f'(\alpha_1) \\ \vdots \\ f^{(m-1)}(\alpha_1) \end{bmatrix}, \begin{bmatrix} f(\alpha_2) \\ f'(\alpha_2) \\ \vdots \\ f^{(m-1)}(\alpha_2) \end{bmatrix}, \cdots, \begin{bmatrix} f(\alpha_n) \\ f'(\alpha_n) \\ \vdots \\ f^{(m-1)}(\alpha_n) \end{bmatrix} \right) \cong \begin{bmatrix} f(\alpha_1) & \cdots & f(\alpha_n) \\ f'(\alpha_1) & \cdots & f'(\alpha_n) \\ \vdots & \ddots & \vdots \\ f^{(m-1)}(\alpha_1) & \cdots & f^{(m-1)}(\alpha_n) \end{bmatrix}}
\end{equation}
\textit{where $f'(X)$ denotes the formal derivative of $f(X)$, and $f^{(i)}(X)$ its $i^{th}$ formal derivative. This codes has length $n$, rate $R=\frac{k}{nm}$ and minimum distance $d_{\text{min}}=n-\floor{\frac{k-1}{m}}\simeq(1-R)n$. Furthermore, for $m=1$ we get a} $\rs_q[n,k]$.
\end{Def}

$\ind$ Consider the received corrupted codeword from $\Der_q[n,k]$ as a string $\yb\in(\F_q^m)^n\cong\F_q^{m\times n}$, which we realize as a $m\times n$ matrix over $\F_q$; as we did for (\ref{rec_matrix_cw}).
Just like in \cref{int_step_subsec_frs}, the goal is to recover all polynomials $f(X)$ of degree $k-1$ whose encoding (\ref{Der_code_eq}) agrees with $\yb$ in at least $t$ columns. This corresponds to decoding from $n-t$ symbol errors for $\Der_q^{(m)}[n,k]$. The algorithm we present, as the one in \cref{FRS_sec}, may be viewed as a higher dimensional analog of the Berlekamp-Welch algorithm. 

\subsection{Interpolation step}
\label{int_step_subsec_der}  

The interpolation step is similar in spirit to the one presented in \cref{int_step_subsec_frs}. Using the same notation, let 
$$ \W=\big\{B_0(X)+B_1(X)Y_1+\cdots+B_m(X)Y_m\mid B_i\in\F_q[X]\big\} $$
which is a $\F_q$-linear subspace $\F_q[X,\Yb]$. For $p(X)\in\F_q[X]$ and $i\in\N_m$, we define the $\F_q$-linear map
$$ D:p(X)\longmapsto p'(X) = \left(B_0'(X)+\sum\limits_{i=1}^mB_i'(X)Y_i\right) \qquad \text{and} \qquad D:p(X)Y_i\longmapsto \Big(p'(X)Y_i+p(X)Y_{i+1}\Big) $$
from $\F_q[X]$ to $\W$, where we take $Y_{m+1}=Y_1$.\\
$\ind$ For $s\in\N_m$, we define the \textit{nonzero} polynomial $Q(X,\Yb)$ as in (\ref{Q_interp}), satisfying the conditions
\begin{equation}
\label{int_cond_der}
  Q(\alpha_i,y_{1i},\cdots,y_{si})=0 \qquad \text{ and } \qquad (D^kQ)(\alpha_i,y_{1i},\cdots,y_{mi})=0
\end{equation}
for all $i\in\N_n$, where $k\in\N_{m-s}$ and $D^k$ denotes the $k$-fold composition of $D$ (apply it $k$ times). Note that conditions (\ref{int_cond_der}) resemble (\ref{int_cond_frs}). Furthermore, note that for each $i$ the conditions (\ref{int_cond_der}) comprises a collection of $(m-s+1)$ homogeneous linear constraints on the coefficients of $Q$.\\
$\ind$ The next two lemmas show why the conditions suffice, and that $Q$ exists and can be found efficiently. The proofs are relatively simple. For the first substitutions $Y_i=f^{(i-1)}(X)$ take place, and for the second it suffices to solve a homogeneous linear system imposed on the coefficients of $Q$ with at most $nm$ constraints. The details can be found in \cite{GW11}. Once again, there is a resemblance between (\ref{s_list_cond}) an (\ref{der_int_lem_eq}).

\begin{Lemma}
\label{der_int_lem_1}
\textit{Suppose $Q$ of the form (\ref{Q_interp}) satisfies (\ref{int_cond_der}). If the received word $\yb$ agrees with the encoding (\ref{Der_code_eq}) at location $i$, i.e. $f^{(j)}(\alpha_i)=y_{j+1,i}$ for $0\leq j<m$ (row $j+1$ of $\yb$), then the polynomial
\begin{equation}
\label{der_int_lem_eq}
  \hat{Q}(X)\coloneqq Q\big(X,f(X),f'(X),\cdots,f^{(s-1)}(X)\big) \quad \text{ satisfies }  \quad \hat{Q}(\alpha_i)=0 \ \text{ and } \ \hat{Q}^{(k)}(\alpha_i)=0 
\end{equation}
for all $k\in\N_{m-s}$, where $\hat{Q}^{(k)}(X)$ is the $k^{th}$ derivative of $\hat{Q}$}.
\end{Lemma}

\begin{Lemma}
\label{der_int_lem_2}
\textit{Let $d$ be as in (\ref{d_express}), except that $N$ is replaced with our current block length $n$. Then, a nonzero $Q$ of the form (\ref{Q_interp}) satisfying (\ref{int_cond_der}), with} deg$(A_0)\leq d+k-1$ \textit{and} deg$(A_i)\leq d$ \textit{for $j\in\N_s$ exists, and can be found in $O\left((nm)^3\right)$ field operations over $\F_q$}.
\end{Lemma}

\subsection{Retrieve candidate polynomials}
\label{retr_cand_pol_subs}

Now that we know how to find a polynomial $Q(X,\Yb)$ satisfying (\ref{int_cond_der}), it remains to list the polynomials $f(X)$ which agree in sufficiently many locations with the received word $\yb$. The following lemma gives an identity which should be satisfied by these candidate polynomials.

\begin{Lemma}
\label{lemma_retr_cand_pols}
\textit{If $f(X)\in\F_q[X]$ has degree at most $k-1$ and an encoding (\ref{Der_code_eq}) agreeing with the received word $\yb$ in at least $t>\frac{d+k-1}{m-s+1}$ columns, then}
\begin{equation}
\label{retr_cand_pols_cond}
  \hat{Q}(X)=Q\big(X,f(X),f'(X),\cdots,f^{(s-1)}(X)\big)=0.
\end{equation}
\end{Lemma}

$\ind$ Lemma \ref{lemma_retr_cand_pols} is identical to lemma \ref{alg_cond_msg}, with the only difference that we substitute $f{(\gamma^{i-1}X)}$ with $f^{(i-1)}(X)$ for all $i\in\N_s$.With our choice of $d$, it follows that any $f(X)$ which agrees with $\yb$ on
$$ t>\frac{d+k-1}{m-s+1}\geq\frac{\left[\frac{n(m-s+1)-k+1}{s+1}+k-1\right]}{m-s+1}=\frac{n}{s+1}+\frac{k-1}{m-s+1}\cdot\left(1-\frac{1}{s+1}\right)=\frac{n}{s+1}+\frac{k-1}{m-s+1}\cdot\frac{s}{s+1} $$
columns satisfies (\ref{retr_cand_pols_cond}). Similarly to \cref{root_find_subsec}, our second step now is to find all polynomials $f(X)$ of degree at most $k-1$, such that
\begin{equation}
\label{A_pol_der}
  \Xi(X) \coloneqq A_0(X)+A_1(X)f(X)+A_2(X)f'(X)+\cdots+A_s(X)f^{(s-1)}(X)=0
\end{equation}
where $A_i(X)=\sum\limits_{j=0}^{\text{deg}(A_i)}a_{ij}X^j$ for each $i$. We view (\ref{A_pol_der}) as a system of linear equations over $\F_q$ in the coefficients of $f(X)=\sum\limits_{i=0}^{k-1}f_iX^i\in\F_q[X]$, for which we note the following fact.

\begin{Fact}
\label{fact_aff_space_two}
\textit{The solutions $(f_0,f_1,\cdots,f_{k-1})$ of} (\ref{A_pol_der}) \textit{form an affine subspace of $\F_q^k$}.
\end{Fact}

$\ind$ The goal is almost identical to the one in \cref{LinAlg_sec}. That is, we want to bound the dimension of the affine subspace of solutions of (\ref{A_pol_der}) by exposing its structure, and then use this to efficiently find an explicit basis. For this, it suffices to give an algorithm in the case that the constant term $a_{s0}$ of $A_s(X)$ is nonzero (\cite{GW11} lemma 5), and we can then use lemma \ref{dim_subsp_der_codes}.

\begin{Lemma}
\label{dim_subsp_der_codes}
\textit{If $a_{s0}\neq0$, the affine solution space of (\ref{A_pol_der}) has dimension at most $s-1$}.
\end{Lemma}

\begin{proof}
The proof idea is parallel to that of lemma \ref{aff_subsp_lem}. The coefficients of $X^r$ of $\Xi(X)$ equals
\begin{align*}
  \xi_r = a_{0r} &+ \big(a_{10}\cdot f_r+a_{11}\cdot f_{r-1}+\cdots+a_{1r}\cdot f_{0}\big)+\big(a_{20}\cdot (r+1)\cdot f_{r+1}+a_{21}\cdot r\cdot f_{r}+\cdots+a_{2r}\cdot f_1\big)+\\
  &+ \cdots + \left(a_{s0}\cdot\frac{(r+s-1)!}{r!}\cdot f_{r+s-1}+\cdots+a_{r1}\cdot(s-1)!\cdot f_{s-1}\right)\\
  = a_{0r} &+ \sum_{j=1}^s\sum_{k=0}^r\frac{(k+j-1)!}{k!}\cdot a_{j(r-k)}\cdot f_{k+j-1} .
\end{align*}
If $(f_0,\cdots, f_{k-1})$ is a solution to $\Xi(X)$, then $\xi_r=0$ for every $r$. For each $r$; $\xi_r$ depends only on $f_j$ for $j<r+s$, and the coefficient of $f_{r+s-1}$ is
$$ a_{s0}\cdot(r+s-1)\cdot(r+s-2)\cdots(r+1) = a_{s0}\cdot\left(\prod_{l=r+1}^{r+s-1}l\right) = a_{s0}\cdot\frac{(r+s-1)!}{r!}. $$
By the assumption that char$(\F_q)>k$, it follows that this coefficient is nonzero when $r+s\leq k$. Hence, if we fix $\{f_i\}_{i=0}^{s-2}$, the rest of the coefficients $\{f_i\}_{i=s-1}^{k-1}$ are uniquely determined. This implies that the dimension of the solution space is at most $s-1$.
\end{proof}

$\ind$ What we showed implies the main result of \cite{GW11} (its theorem 6 and corollary 7), which for parameters $s\simeq1/\varepsilon,m\simeq s^2$ suggests that $\Der_q^{(m)}[n,k]$ of rate at least $R$ for $R\in(0,1)$, can be list decoded from a fraction of $1-R-\varepsilon$ of errors, with a list-size of $q^{O(1/\varepsilon)}$. Lastly, there is potential for improving the large list-size of derivative codes, by drawing codewords from subspace-evasive sets.

\section{Concluding Remarks}
\label{Concl_rmks_sec}

In this survey we first saw how the gap was closed for the optimal trade-off between rate and error-correction capability for list decoding algorithmically, through folded Reed-Solomon codes. We then showed several ways in which this can be achieved, using very different approaches, but at their core same ideas, and attaining same results. This is not just impressive, but also important; as different point of views may clear any ambiguity and make things easier to interpret and understand. We also looked into local self-correction and local decodability of multiplicity codes.\\
$\ind$ Two list decoding algorithms which were not discussed, are the list decoding of Parvaresh-Vardy (PV) codes \cite{PV05} which has decoding radius $1-\sqrt[M+1]{M^MR^M}$ for an arbitrary parameter $M\in\Z_+$, and the list decoding of multiplicity codes \cite{Kop15} which achieves the list decoding capacity $1-R$. Chronologically, the first major breakthrough in this area was presented in \cite{Sud97} which had radius $1-2\sqrt{R}$, followed by the Guruswami-Sudan radius $1-\sqrt{R}$ \cite{GS98}, and then PV was a stepping stone between towards folded Reed-Solomon codes. The main ideas in all these achievements come from the seminal paper of Sudan, and a lot of what we presented relates to the construction of PV codes; in which certain powers of the evaluations of the polynomial $f(X)$ are being bundled together. These codes have also been used in other applications, e.g. randomness extractors.\\
$\ind$ There are a lot more articles in this (general) area which were not discussed. We only bring to your attention two such articles. In \cite{Gur09}, the folding operation was extended to certain algebraic-geometry codes, which contain FRS is a special case. These codes are referred to as \textit{folded cyclotomic codes}. The second is \cite{HN09}, which uses a similar approach to what was discussed in \cref{HL_sec}, for folded versions of algebraic-geometric codes.


\bibliographystyle{alpha}
\bibliography{refs.bib}

\appendix

\section{Digression Into Number Theory --- $p$-adic numbers}
\label{NT_subs}

We digress from coding theory in this appendix to further discuss Hensel's lemma from \cref{HL_sec}, and its connection to the $p$-adic numbers $\Q_p$. In 1897, Kurt Hensel himself introduced the field of $p$-adic numbers $\Q_p$, which have been thoroughly studied throughout the $20^{th}$ century and are still are active research area, though they were foreshadowed in Ernst Kummer's work a few decades earlier. The first major breakthrough involving $p$-adic numbers is the \textit{Hasse–Minkowski} theorem, which can be used to test efficiently whether a Quadratic form has a solution in $\Q$. In the literature, there are also examples of codes over the $p$-adic integers and numbers; e.g. \cite{CS95},\cite{DHP06}.  \\
$\ind$ Vaguely speaking, they allow the use of analytic methods in the study of Diophantine equations, number theory, arithmetic geometry and more recently, numerical analysis \cite{Car17}. After all, Hensel's main motivation was the analogy between the unique factorization domain (UFD) $\Z$ along with its field of fractions $\Q$, and the UFD $\C[X]$ along with its field of fractions $\C(X)$. Essentially, $p\in\Z$ are analogous to the (irreducible) polynomials $(X-a)\in\C[X]$ \cite{Ogg14}. Here is a definition of the $p$-adic integers $\Z_p$, and two definitions of the $p$-adic numbers $\Q_p$; an algebraic \ref{def_Qp_1} and an analytic \ref{def_Qp_2} (which may be viewed as a theorem). By $p$ we indicate a fixed prime.

\begin{Def}
\label{def_Zp}
\textit{A \textbf{$p$-adic integer} is a formal sum $\alpha=\sum\limits_{i=0}^{\infty}a_ip^i$, for integers $0\leq a_i<p$. The set of $p$-adic integers $\Z_p$, forms a commutative ring. We can alternatively write $\alpha=\cdots a_i\cdots a_2a_1a_0$}.
\end{Def}

\begin{Def}
\label{def_Qp_1}
\textit{The \textbf{$p$-adic numbers} are the series of the form
$$ a_{-n}\frac{1}{p^n}+a_{-n+1}\frac{1}{p^{n-1}}+\cdots+a_{-1}\frac{1}{p}+a_0+a_1p+\cdots $$
which form the field we denoted by $\Q_p$. Furthermore $\Q\subseteq\Q_p$, and if $\alpha\in\Q_p$, then $\exists N\geq0$ such that $p^N\alpha\in\Z_p$. In other words, $\Q$ may be viewed as a subfield of $\Q_p$}.
\end{Def}

\begin{Def}
\label{def_p_adic_val}
\textit{Let $\alpha\in\Q^{\times}$ where $\alpha=p^k\frac{g}{h}$ for $k\in\Z$, and $p,g,h$ coprime. The \textbf{$p$-adic valuation} of $\alpha$ is} $\text{ord}_p(\alpha)=k$ \textit{and its \textbf{$p$-adic absolute value} is $\nu_p(\alpha)=p^{-k}$, which is a non-Archimedean metric; as $\nu_p(\beta+\gamma)\leq\max\left\{\nu_p(\beta),\nu_p(\gamma)\right\}$ for $\beta,\gamma\in\Q$. By convention} $\text{ord}_p(0)=\infty$ \textit{and $\nu_p(0)=0$}.
\end{Def}

\begin{Def}
\label{def_Qp_2}
\textit{The field of \textbf{$p$-adic numbers} $\Q_p$ is the completion of $\Q$ with respect to the metric induced by $\nu_p(\cdot)$, i.e. every Cauchy sequence converges. Moreover, $\Q$ is dense in $\Q_p$ (as is $\Z$ in $\Z_p$)}.
\end{Def}

$\ind$ By definition $\Q_p=\Z_p[\frac{1}{p}]$, and it is the fraction field of $\Z_p$. Another definition which resembles Hensel's lemma, is defined through the ring homomorphism
\begin{align*}
  \pi_n \ : \quad \Z_p \qquad &\longrightarrow \ \Z/p^n\Z \\
  \sum\limits_{i=0}^{\infty}a_ip^i \ &\longmapsto \ \left(\sum\limits_{i=0}^{n-1}a_ip^i\right) \bmod p^n
\end{align*}
for which $\pi_{n+1}(\alpha)\equiv\pi_n(\alpha)\bmod p^n$. This definition uses the \textit{projective/inverse limit}, and is not relevant to what we want to show. We want to demonstrate the resemblance with Hensel's lemma.

\begin{Lemma}[Basic version \cite{ConHL}]
\label{HL_basic}  
\textit{If $f(X)\in\Z_p[X]$ and $a\in\Z_p$ satisfies $f(a)\equiv0\bmod p$ and $f'(a)\not\eq0\bmod p$, then there exists a unique $\alpha\in\Z_p$ such that $f(\alpha)=0$ and $\alpha\equiv a\bmod p$}.
\end{Lemma}

\begin{Thm}[Stronger version \cite{ConHL}]
\label{HL_stronger}  
\textit{Let $f(X)\in\Z_p[X]$ and $a\in\Z_p$ satisfy $\nu_p(f(a))<\nu_p(f'(a))^2$. Then, there is a unique $\alpha\in\Z_p$ such that $f(\alpha)=0$ and $\nu_p(\alpha-a)<\nu_p(f'(a))$. Moreover:}
$$ (1) \ \ \nu_p(\alpha-a) = \nu_p\big(f(a)/f'(a)\big)<\nu_p(f'(a)) \qquad 
\text{ and } \qquad (2) \ \ \nu_p(f'(\alpha))=\nu_p(f'(a)). $$
\end{Thm}

$\ind$ One can restate the above theorem in a way which gives a construction of the $\alpha\in\Z_p$ \cite{Car17}. The striking part about this statement (and the construction of $\alpha$), is that the proof applies Newton's method; establishing connections now to numerical analysis. This is a (approximate) root-finding algorithm, which takes us back to \cref{root_find_subsec}. The remarkable thing about Newton's method is that it extends almost word for word to Hensel's lemma \ref{HL_stronger}, when $\R$ is replaced by $\Q_p$. More precisely, under the assumptions of lemma \ref{HL_stronger} we
construct the sequence ${(x_i)}_{i\in\N_0}$ by the recurrence $x_0=a$; $x_{i+1}=x_i-f(x_i)/f'(x_i)$, which converges to $\alpha\in\Z_p$ with $f(\alpha)=0$.\\
$\ind$ The $p$-adics are relatively hard to grasp and understand, though they have ``simple'' constructions (e.g. lemma \ref{Hens_lem_simple}). Part of the reason is that there are many ways to interpret them, as we have seen. We briefly discuss a final more visual representation of $\Z_p$, which is more meaningful and convenient geometrically. From definition \ref{def_Zp}, it is clear that any $\alpha\in\Z_p$ can be decomposed in base $p$. We can then construct a tree with $p$ branches at each node (a \textit{full $p$-ary tree}), with each branch corresponding to an integer coefficient $0\leq a_i<p$, and nodes at the same \textit{height} $h$ correspond to elements of the congruence class $\bmod p^{h+1}$ (height here corresponds to the depth of the tree). Where Where does Hensel's lemma come into play? We point out that definitions \ref{Hens_lem_simple} and \ref{HL_basic} are in fact the same, with the latter stated in a more abstract way.\\
$\ind$ Name anything $p$-adic and most likely it has already been well-defined and studied extensively, from $p$-adic differential equations \cite{Ked10}, to $p$-adic modular forms \cite{Gou06} and $p$-adic $\zeta$-functions \cite{Kob12}. The most common use of $p$-adics though, is probably in the study of elliptic curves. This is where they appear in the solution of one of the most important problems in mathematics, Fermat's last theorem (specifically, the proof of the modularity conjecture for semistable elliptic curves). As a humbled mathematician said twenty-six years ago, `I think I'll stop here'.


\end{document}